\documentclass[journal]{IEEEtran}
\usepackage{amsmath}
\usepackage{amssymb}
\usepackage{amsthm}
\usepackage{graphicx}
\usepackage{epstopdf}
\usepackage{cite}
\usepackage{algorithm}
\usepackage{algpseudocode}
\usepackage{tikz}
\usepackage{pgfplots}
\usepackage{mysymbol}
\usepackage{subfigure}
\usepgfplotslibrary{groupplots}

\IEEEoverridecommandlockouts

\theoremstyle{plain}
\newtheorem{theorem}{Theorem}
\newtheorem{corollary}[theorem]{Corollary}
\newtheorem{lemma}[theorem]{Lemma}
\newtheorem{proposition}[theorem]{Proposition}

\theoremstyle{definition}

\theoremstyle{remark}
\newtheorem{remark}[theorem]{Remark}

\newcommand{\R}{\mathbb{R}}
\begin{document}

\title{Stochastic Routing and Scheduling Policies for Energy Harvesting Communication Networks}
\author{Miguel Calvo-Fullana, Carles Ant\'on-Haro, Javier Matamoros, and Alejandro Ribeiro
\thanks{
This work is supported by ARL DCIST CRA W911NF-17-2-0181 and the Intel Science and Technology Center for Wireless Autonomous Systems.
%This work was partially supported by the Catalan government under grant SGR2014-1567; the Spanish government under grant TEC2013-44591-P (INTENSYV), and grant PCIN-2013-027 (E-CROPS) in the framework of the ERA-NET CHIST-ERA program.
}
%\thanks{
%This work was partially supported by the Catalan government under grant SGR2014-1567; the Spanish government under grant TEC2013-44591-P (INTENSYV), and grant PCIN-2013-027 (E-CROPS) in the framework of the ERA-NET CHIST-ERA program.
%}
\thanks{
M. Calvo-Fullana, and A. Ribeiro are with the Department of Electrical and Systems Engineering, University of Pennsylvania, Philadelphia, PA 19104, USA (e-mail: \mbox{cfullana}@seas.upenn.edu; \mbox{aribeiro}@seas.upenn.edu).
}
\thanks{
C. Ant\'on-Haro, and J. Matamoros are with the Centre Tecnol\`ogic de Telecomunicacions de Catalunya (CTTC/CERCA), 08860 Castelldefels, Barcelona, Spain (e-mail: \mbox{carles.anton}@cttc.cat; \mbox{javier.matamoros}@cttc.cat).
}
%\thanks{
%M. Calvo-Fullana, C. Ant\'on-Haro, and J. Matamoros are with the Centre Tecnol\`ogic de Telecomunicacions de Catalunya (CTTC/CERCA), 08860 Castelldefels, Barcelona, Spain (e-mail: \mbox{miguel.calvo}@cttc.cat; \mbox{carles.anton}@cttc.cat; \mbox{javier.matamoros}@cttc.cat).
%}
%\thanks{
%A. Ribeiro is with the Department of Electrical and Systems Engineering, University of Pennsylvania, Philadelphia, PA 19104 USA (e-mail: \mbox{aribeiro}@seas.upenn.edu).
%}
\thanks{
This work has been presented in part at the 2017 IEEE International Conference on Acoustics, Speech and Signal Processing (ICASSP)\cite{calvo2017stochastic}.
}
}

\maketitle

\begin{abstract}
In this paper, we study the joint routing-scheduling problem in energy harvesting communication networks. Our policies, which are based on stochastic subgradient methods on the dual domain, act as an energy harvesting variant of the stochastic family of backpresure algorithms. Specifically, we propose two policies: (i) the Stochastic Backpressure with Energy Harvesting (SBP-EH), in which a node's routing-scheduling decisions are determined by the difference between the Lagrange multipliers associated to their queue stability constraints and their neighbors'; and (ii) the Stochastic Soft Backpressure with Energy Harvesting (SSBP-EH), an improved algorithm where the routing-scheduling decision is of a probabilistic nature. For both policies, we show that given sustainable data and energy arrival rates, the stability of the data queues over all network nodes is guaranteed. Numerical results corroborate the stability guarantees and illustrate the minimal gap in performance that our policies offer with respect to classical ones which work with an unlimited energy supply.
\end{abstract}

\begin{IEEEkeywords}
Energy harvesting, wireless networks, backpressure, routing, scheduling.
\end{IEEEkeywords}

\IEEEpeerreviewmaketitle

\section{Introduction}
\label{sec:Introduction}

Providing wireless devices with Energy Harvesting (EH) capabilities enables them to acquire energy from their surroundings. The sources from which to obtain such energy can be of a varied nature, with some of the most common being thermal, vibrational or solar sources \cite{vullers2010energy}. Such ample variety of energy sources, coupled with recent hardware advancements, enables devices to acquire sufficient energy to power themselves. This, in turn, frees these devices from the constraints that traditional battery-only operation imposes. Nonetheless, the random and intermittent nature of this new energy supply calls for a new approach to the design of communication policies. 

As a consequence, there is significant interest in the study of communication devices powered by energy harvesting. The scenarios of EH-aware communication studied in the literature are vast and range from throughput maximization \cite{yang2012optimal,tutuncuoglu2012optimum,ozel2011transmission,ho2012optimal,ozel2014optimal}, source-channel coding \cite{calvo2017reconstruction,castiglione2014energy,orhan2014source,castiglione2012energy}, estimation \cite{yang2013wireless,knorn2015distortion,calvo2016sensor}, simultaneous information and power transfer \cite{huang2013simultaneous,xu2014multiuser,zheng2014information,liu2014secrecy} and many others (see \cite{ulukus2015energy} for an overall review of current research efforts). 

The appearance of multiple interconnected devices powered by energy harvesting results in communication networks formed by self-sustainable and perpetually communicating nodes. In such scenarios, there is the necessity of designing efficient routing and scheduling algorithms that explicitly take into account the energy harvesting process. In this sense, there have been some previous efforts in developing communication policies for these types of multi-hop networks. In general, the full characterization of the optimal transmission policies is a difficult problem, as optimal transmission policies are heavily coupled throughout the network. Under full non-causal knowledge of the energy harvesting process, the optimal transmission policies of a simpler two-hop network have been studied in \cite{orhan2015energy}. A more realistic approach is the consideration of non-causal knowledge of the energy harvesting process. Under this assumption, the authors in\cite{tapparello2014dynamic} jointly optimize data compression and transmission tasks to obtain a close-to-optimal policy. In \cite{lin2007asymptotically}, the authors propose an EH-aware routing scheme that is asymptotically optimal with respect to the network size. The authors in \cite{gatzianas2010control}, address the EH scheduling problem for both single-hop and multi-hop networks, and provide a joint admission control and routing policy. Also, in the same line, the authors in \cite{huang2013utility} propose a policy which improves on the multi-hop performance bounds of \cite{gatzianas2010control}. Overall, non-causal policies are typically designed under the assumption of independent and identically distributed (i.i.d.) or Markov energy harvesting and data arrival processes, and Lyapunov optimization techniques are used to derive their queue stability results.

In this paper, we study the problem of jointly routing and scheduling data packets in an energy harvesting communication network. We start by introducing the system model in Section \ref{sec:SystemModel}. We consider a communication network where each node independently generates traffic for delivery to a specific destination and collaborates with the other nodes in the network to ensure the delivery of all data packets. In this way, each node decides the next suitable hop for each packet in its queue (routing), and when to transmit it (scheduling). The solution to this problem\textemdash when the nodes are not EH-powered\textemdash is given by the backpressure (BP) algorithm \cite{tassiulas1992stability}. When the nodes are powered by energy harvesting, the previous works \cite{gatzianas2010control} and \cite{huang2013utility} considered a similar problem, which consists in finding admission control and resource allocation policies that satisfy network stability and energy causality while attaining close-to-optimal performance. In our work, instead, the goal is to find stabilizing policies given the data rates. Also, while previous works \cite{gatzianas2010control,huang2013utility} require data and energy arrival processes to be i.i.d. or Markov, we only require them to be ergodic, which is a weaker requirement. Furthermore, our approach to the problem is also markedly different. While the works \cite{gatzianas2010control} and \cite{huang2013utility} relied on queueing theory and Lyapunov drift arguments to find stabilizing policies, we instead interpret the scheduling and routing problem as a stochastic optimization problem. This allows us to resort to a dual stochastic subgradient descent algorithm \cite{ribeiro2010ergodic} to solve the joint routing-scheduling problem. 

 We devote Section \ref{sec:Alg} to the development of the proposed stochastic joint routing and scheduling algorithms. The main issue to tackle is that the introduction of energy harvesting constraints results in a causality problem regarding the energy consumption. In order to solve this, we introduce a modified problem formulation that allows us to ensure causality. Under this framework, we propose two different policies. The first, which we denote Stochastic Backpressure with Energy Harvesting (SBP-EH), is a policy of rather simple nature. The network nodes track the pressure of the data flows by computing the difference between the Lagrange multipliers associated to their queue stability constraints and the ones of their neighbors (instead of their data queues as in the classical backpressure algorithm). Then, the Lagrange multipliers associated with the battery state reduce the pressure when the stored energy in the node decreases. The resulting routing-scheduling decision is to transmit the flow with highest pressure. The second policy, which we name Stochastic Soft Backpressure with Energy Harvesting (SSBP-EH), is a probabilistic policy. In this policy, the nodes perform the same tracking of pressure as the SBP-EH policy. However, instead of transmitting the flow with the highest pressure, the flows are equalized in an inverse waterfilling manner. This results in a routing-scheduling probability mass function, where the transmit decision is taken as a sample of this distribution. This second policy, while not as simple as the previous one, provides several improvements in the stabilization speed of the network, as well as a reduction in the packets in queue and packet delivery delay once the network is stabilized. 

Theoretical guarantees, namely, queue stability and energy causality are discussed in Section \ref{sec:Stability}. For both policies, we provide the necessary battery capacity which ensures the proper behavior of the algorithms. Furthermore, we also certify that given sustainable data and energy arrival rates, the stability of the data queues over all network nodes is guaranteed. After this, we dedicate Section \ref{sec:NumericalResults} to simulations assessing the performance of our proposed policies and verify that they show a minimal gap in performance with respect to classical policies operating with an unlimited energy supply. Finally, we provide some concluding remarks in Section \ref{sec:Conclusions}.

\section{System Model}
\label{sec:SystemModel}

Consider a communication network given by the graph $\mathcal{G}=(\mathcal{N},\mathcal{E})$, where $\mathcal{N}$ is the set of $N$ nodes in the network and $\mathcal{E} \subseteq \mathcal{N} \times \mathcal{N}$ is the set of communication links, such that if node $i$ is capable of communicating with node $j$, we have $(i,j) \in \mathcal{E}$. Moreover, we define the neighborhood of node $i$ as the set $\mathcal{N}_i=\{j | (i,j)\in \mathcal{E}\}$. The network supports $K$ information flows (which we index by the set  $\mathcal{K}$), where for a flow $k \in \mathcal{K}$, the destination node is denoted by $N^k_{(dest)}$. At a time slot $t$, each $k \in \mathcal{K}$ flow at the $i$-th node generates $a_i^k[t]$ packets to be delivered to the node $N^k_{(dest)}$. This packet arrival process is assumed to be stationary with mean $\mbE \bigl[ a^k_i[t] \bigr] = a^k_i$. At the same time, the $i$-th node routes $r_{ij}^k[t]$ packets to its neighbors  $j \in \mathcal{N}_i$, while simultaneously being routed $r_{ji}^k[t]$ packets. For simplicity, at each time slot, we restrict each node to route one single packet to its neighbors. Therefore, the nodes have the following routing constraint
\begin{align}
\sum_{k \in \mathcal{K}} \sum_{j \in \mathcal{N}_i} r_{ij}^k[t]\leq 1, \quad i \in \mathcal{N}.
\label{eq:routingCap}
\end{align}
Furthermore, each node in the network keeps track of the number of packets awaiting to be transmitted for each flow. Denoting by $q_i^k[t]$ the $k$-th flow data queue at the $i$-th node and time slot $t$, the evolution of the queue is given by
\begin{align}
q_i^k[t+1]=q_i^k[t]+a_i^k[t]+\sum_{j \in \mathcal{N}_i} r_{ji}^k[t]-\sum_{j \in \mathcal{N}_i} r_{ij}^k[t],\label{eq:queueEvol}
\end{align}
for all $i \in \mathcal{N}$ and $k \in \mathcal{K}$. The objective is to determine routing policies $r_{ij}^k[t]$ such that the queues in \eqref{eq:queueEvol} remain stable while satisfying the routing constraints given by \eqref{eq:routingCap}. By grouping all the queues in a vector $\mathbf{q}[t]=\{q_i[t]\}$, we say that the routing policies $r_{ij}^k[t]$ guarantee stability if there exists a constant $Q$ such that for some arbitrary time $T$ we have 
\begin{align} 
\Pr { \max_{t \geq T} \|\mathbf{q}[t]\| \leq Q \given \mathbf{q}[T] } = 1.
\label{eq:queueStability}
\end{align}
This is to say that, almost surely, no queue becomes arbitrarily large. In turn, we can guarantee this if the average rate at which packets enter the queues is smaller than the rate at which they exit them. In order to formally state this, let us denote the ergodic limits of processes $a^k_i[t]$ and $r^k_{ij}[t]$ by
\begin{align} 
a_i^k &=\mbE \bigl[ a^k_i[t] \bigr] = \lim_{t \to \infty}\frac{1}{t} \sum\limits_{l=1}^{t} a_i^k[l], \\
r_{ij}^k &= \mbE \bigl[ r_{ij}^k [t] \bigr] =  \lim_{t \to \infty}\frac{1}{t} \sum\limits_{l=1}^{t} r_{ij}^k[l].
\end{align}
Then, in order to have stable data queues in the network, it suffices to satisfy the condition
\begin{align} 
a^k_i < \sum_{j \in \mathcal{N}_i} r_{ij}^k - \sum_{j \in \mathcal{N}_i} r_{ji}^k.
\label{eq:queueStabilityStrictInequality}
\end{align}
for all $i \in \mathcal{N}$ and $k \in \mathcal{K}$. If there exist routing variables $r_{ij}^k$ satisfying this inequality, then the queue evolution in \eqref{eq:queueEvol} follows a supermartingale expression, and the stability condition given by \eqref{eq:queueStability} is then guaranteed by the martingale convergence theorem \cite[Theorem 5.2.9]{durrett2010probability}. Alternatively, by introducing arbitrary concave functions $f_{ij}^k : \R \to \R$, we can formulate this as the following optimization problem
\begin{subequations}
\begin{align}\label{eq:optProblemNoEH}
   \underset{\sum\limits_{k,j} r_{ij}^k \leq 1}{\text{maximize}}  \quad & \sum_{i \in \mathcal{N}} \sum_{k \in \mathcal{K}} \sum_{j \in \mathcal{N}_i}  f_{ij}^k \left( r_{ij}^k \right) \\
   \text{subject to}
		\quad   & a^k_i \leq \sum_{j \in \mathcal{N}_i} r_{ij}^k - \sum_{j \in \mathcal{N}_i} r_{ji}^k, \quad k\in\mathcal{K},i\in\mathcal{N},
					\label{eq:optProblemConstraintQueueNoEH} 
\end{align}
\end{subequations}
Observe that in \eqref{eq:optProblemConstraintQueueNoEH} the inequality allows for equality whereas in \eqref{eq:queueStabilityStrictInequality} the inequality is strict. This mismatch is necessary because optimization problems are not well behaved on open sets. We can then think of \eqref{eq:optProblemConstraintQueueNoEH} as a relaxation of \eqref{eq:queueStabilityStrictInequality} but one of little practical consequence as it is always possible to add a small slack term to \eqref{eq:optProblemConstraintQueueNoEH} to produce a non-strict inequality that implies the strict inequality in \eqref{eq:queueStabilityStrictInequality}. We don't do that to avoid a cumbersome term of little conceptual value. We emphasize that implicit to \eqref{eq:routingCap} is the constraint $\sum_{k \in \mathcal{K}} \sum_{j \in \mathcal{N}_i} r_{ij}^k \leq 1$ for all $i \in \mathcal{N}$, which is the same as \eqref{eq:routingCap} but for average variables. We will ensure later that the algorithm we design satisfies the constraint not just on average but for all time instances -- see Section \ref{sec:Alg}.

Assuming data arrival rates satisfying $\sum_{k \in \mathcal{K}} \sum_{j \in \mathcal{N}_i} r_{ij}^k \leq 1$ for all $i \in \mathcal{N}$ as well as the inequalities in \eqref{eq:optProblemConstraintQueueNoEH} exist, the objective is to design an algorithm such that the instantaneous routing variables $r_{ij}^k[t]$ satisfy $\mbE \bigl[ r^k_{ij}[t] \bigr]=r^k_{ij}$ and the routing constraints in \eqref{eq:routingCap} are satisfied for all time slots. This is the optimization problem that the backpressure family of algorithms solve. By resorting to a stochastic subgradient method on the dual domain, a direct comparison can be established between data queues and Lagrange multipliers \cite{huang2011delay}. Then, the choice of objective function in the optimization problem \eqref{eq:optProblemNoEH} determines the resulting variant of the backpressure algorithm. For example, on one hand, the stochastic backpressure (SBP) algorithm \cite{tassiulas1992stability} can be recovered by the use of a linear objective function. On the other hand, the choice of a strongly concave objective function leads to the soft stochastic backpressure (SSBP) algorithm \cite{ribeiro2009stochastic}.

\subsection{Routing and Scheduling with Energy Harvesting}

Different from classical approaches \cite{tassiulas1992stability,ribeiro2009stochastic}, we consider that the network nodes are powered by energy harvesting. At time slot $t$, the $i$-th node harvests $e_i[t]$ units of energy, where the energy harvesting process is assumed to be stationary with mean $\mbE \bigl[e_i[t] \bigr] = e_i$. We consider a normalized energy harvesting process, where the routing of one packet consumes one unit of energy. Furthermore, we consider packet transmission to be the only energy-consuming action taken by the nodes. Under these conditions and denoting by $b_i[t]$ the energy stored in the $i$-th node's battery at time $t$, the following energy causality constraint must be satisfied for all time slots
\begin{align}
\sum_{k \in \mathcal{K}} \sum_{j \in \mathcal{N}_i} r_{ij}^k[t]  \leq b_i[t], \quad i \in \mathcal{N}.
\label{eq:energyCausality}
\end{align}
Additionally, we consider that nodes have a finite battery of capacity $b_i^{\max}$. Then, we can write the battery dynamics as
\begin{align}
b_i[t+1]=\biggl[b_i[t]-\sum_{k \in \mathcal{K}} \sum_{j \in \mathcal{N}_i} r_{ij}^k[t]+e_i[t]\biggr]_{0}^{b_i^{\max}}
\label{eq:batteryEvol}
\end{align}
for $i \in \mathcal{N}$, where $[\cdot]_{0}^{b_i^{\max}}$ denotes the projection to the interval $[0,b_i^{\max}]$. In order to introduce these constraints into the optimization problem \eqref{eq:optProblemNoEH}, we denote the ergodic limit of the energy harvesting process $e_i[t]$ by
\begin{align} 
e_i =\mbE \bigl[ e_i[t] \bigr] = \lim_{t \to \infty}\frac{1}{t} \sum\limits_{l=1}^{t} e_i[l].
\end{align}
Then, substituting the battery dynamics given by \eqref{eq:batteryEvol} in the energy causality constraint \eqref{eq:energyCausality} and then taking the ergodic limits on both sides of the inequality, we obtain the following average constraint in the routing variables 
\begin{align}
\sum_{k \in \mathcal{K}} \sum_{j \in \mathcal{N}_i} r_{ij}^k  \leq e_i, \quad i \in \mathcal{N}.
\end{align}
This states that the average amount of energy spent must be less than the average energy harvested. Then, we introduce this constraint into problem \eqref{eq:optProblemNoEH}, resulting in the following optimization problem
\begin{subequations}
\begin{align}
   \underset{\sum\limits_{k,j} r_{ij}^k \leq 1}{\text{maximize}}  \quad & \sum_{i \in \mathcal{N}} \sum_{k \in \mathcal{K}} \sum_{j \in \mathcal{N}_i}  f_{ij}^k \left( r_{ij}^k \right) \\
   \text{subject to}
		\quad   & a^k_i \leq \sum_{j \in \mathcal{N}_i} r_{ij}^k - \sum_{j \in \mathcal{N}_i} r_{ji}^k, \quad k\in\mathcal{K},i\in\mathcal{N} 
					\label{eq:optProblemConstraintQueueOriginal} \\
		\quad   & \sum_{k \in \mathcal{K}} \sum_{j \in \mathcal{N}_i} r_{ij}^k \leq e_i, \quad i\in\mathcal{N}.
					\label{eq:optProblemConstraintEnergyOriginal} 
\end{align}
\label{eq:optProblemOriginal}
\end{subequations}
Assuming data and energy arrival rates satisfying \eqref{eq:optProblemConstraintQueueOriginal} and \eqref{eq:optProblemConstraintEnergyOriginal} exist, the goal is to design an algorithm such that the instantaneous routing variables $r_{ij}^k[t]$ satisfy $\mbE \bigl[ r^k_{ij}[t] \bigr]=r^k_{ij}$ and the constraints \eqref{eq:routingCap} and \eqref{eq:energyCausality} are satisfied for all time slots. However, the use of the average energy constraint \eqref{eq:optProblemConstraintEnergyOriginal} presents a causality problem, as a solution satisfying \eqref{eq:optProblemConstraintEnergyOriginal} does not guarantee that the energy causality constraint in \eqref{eq:energyCausality} is satisfied for all time slots. In order to circumvent this, we propose the introduction of the following modified optimization problem
\begin{subequations}
\begin{align}
   \underset{
   %\sum\limits_{k,j} r_{ij}^k \leq 1, x_i^k \in [0,\bar{x}_i^k]
   \begin{subarray}{l}
  		\sum\limits_{k,j} r_{ij}^k \leq 1,\\
  		x_i^k \in [0,\bar{x}_i^k]
  	\end{subarray}   
   }{\text{maximize}}  \quad & \sum_{i \in \mathcal{N}} \sum_{k \in \mathcal{K}} \sum_{j \in \mathcal{N}_i}  f_{ij}^k \left( r_{ij}^k \right)
        - \sum_{i \in \mathcal{N}} \sum_{k \in \mathcal{K}} \bar{\gamma}_i^{k} x_i^k \\
   \text{subject to}
		\quad   & a^k_i \leq \sum_{j \in \mathcal{N}_i} \bigl(r_{ij}^k - r_{ji}^k\bigr) + x_i^k, k\in\mathcal{K},i\in\mathcal{N} 
					\label{eq:optProblemConstraintQueue} \\
		\quad   & \sum_{k \in \mathcal{K}} \sum_{j \in \mathcal{N}_i} r_{ij}^k \leq e_i, \quad i\in\mathcal{N}
					\label{eq:optProblemConstraintEnergy} 
\end{align}
\label{eq:optProblem}
\end{subequations}
This optimization problem differs from \eqref{eq:optProblemOriginal} in the introduction of an auxiliary variable $x_i^k$. This variable is restricted to lie in the interval $[0,\bar{x}_i^k]$, with $\bar{x}_i^k$ being a constant whose value is determined by the system parameters. This auxiliary variable appears in the queue stability constraint \eqref{eq:optProblemConstraintQueue}, where it helps to satisfy the constraint if necessary. Furthermore, we have added the term $-\sum_{i \in \mathcal{N}} \sum_{k \in \mathcal{K}} \bar{\gamma}_i^{k}  x_i^k$ in the objective function, where $\bar{\gamma}_i^{k}$ is a constant parameter. The value of this parameter $\bar{\gamma}_i^{k}$ is chosen such that the optimal value of the Lagrange multipliers of the queue constraint \eqref{eq:optProblemConstraintQueue} lies in the interval $[0,\bar{\gamma}_i^{k}]$. We later show in Section \ref{sec:Stability} that this allows us to satisfy the energy causality constraints while also stabilizing the data queues.

\section{Joint Routing and Scheduling Algorithm}
\label{sec:Alg}

As we mentioned previously, in order to solve optimization problem posed in \eqref{eq:optProblem} we resort to a primal-dual method. To start, let us define the vector $\mathbf{r}=\{r_{ij}^k,x_i^k\}$ collecting the routing variables $r_{ij}^k$ and auxiliary variables $x_i^k$ and the vector $\boldsymbol{\lambda}=\{\gamma_i^k,\beta_i\}$ collecting the queue multipliers $\gamma_i^k$ associated with constraint \eqref{eq:optProblemConstraintQueue} and battery multipliers $\beta_i$ corresponding to constraint \eqref{eq:optProblemConstraintEnergy}. Furthermore, we collect the implicit optimization constraints in the set $\mathcal{R}=\bigl\{\sum_{k,j} r_{ij}^k \leq 1, x_i^k \in [0,\bar{x}_i^k]\bigr\}$. Then, we write the Lagrangian of the optimization problem \eqref{eq:optProblem} as follows
\begin{align}
\mathcal{L}(\mathbf{r},\boldsymbol{\lambda})&=\sum_{i \in \mathcal{N}} \sum_{k \in \mathcal{K}}  \sum_{j \in \mathcal{N}_i} f_{ij}^k \left( r_{ij}^k \right) 
	- \sum_{i \in \mathcal{N}} \sum_{k \in \mathcal{K}} \bar{\gamma}_i^{k} x_i^k \nonumber\\
	&+\sum_{k \in \mathcal{K}}   \sum_{i \in \mathcal{N}} \gamma_i^k\biggl(\sum_{j \in \mathcal{N}_i} r_{ij}^k - \sum_{j \in \mathcal{N}_i} r_{ji}^k + x_i^k  - a^k_i \biggr) \nonumber\\
	&+\sum_{i \in \mathcal{N}} \beta_i\biggl(e_i - \sum_{k \in \mathcal{K}} \sum_{j \in \mathcal{N}_i} r_{ij}^k \biggr).
\label{eq:Lagrangian}
\end{align}
The Lagrange dual function is then given by
\begin{align}
g(\boldsymbol{\lambda})=\max_{\mathbf{r} \in \mathcal{R}} \mathcal{L}(\mathbf{r},\boldsymbol{\lambda}).
\label{eq:dualFunction}
\end{align}
An immediate issue that arises when trying to solve this problem is that network nodes have no knowledge of the data arrival rates $a_i^k$ nor the energy harvesting rates $e_i$. Nonetheless, the nodes observe the instantaneous rates $a_i^k[t]$ and $e_i[t]$, hence we resort to using these instantaneous variables. Furthermore, we can reorder the Lagrangian \eqref{eq:Lagrangian} to allow for a separate maximization over network nodes, where each node only needs the queue multipliers of its neighboring nodes. The routing variables can then be obtained as follows
 \begin{align}
r^k_{ij} [t]:= \argmax_{\sum\limits_{k,j} r_{ij}^k \leq 1} 
 \sum_{k \in \mathcal{K}} & \sum_{j \in \mathcal{N}_i}  \biggl( f_{ij}^k \left( r_{ij}^k \right)  \nonumber \\
 & + r_{ij}^k \left(  \gamma_i^k[t] - \gamma_j^k[t] -\beta_i[t] \right)\biggr),
\label{eq:dualFunctionMax}
\end{align}
for $i\in\mathcal{N}$. In a similar way, the auxiliary variables at each node are given by
 \begin{align}
x_i^k [t] := \argmax_{x_i^k \in [0,\bar{x}_i^k]} x_i^k \left( \gamma_i^k[t] - \bar{\gamma}_i^{k}\right).
\label{eq:dualFunctionMaxAuxiliary}
\end{align}
This is simply a threshold operation, where $x_i^k [t]=0$ if $\gamma_i^k[t] \leq \bar{\gamma}_i^{k}$ and $x_i^k [t]=\bar{x}_i^k$ if $\gamma_i^k[t] > \bar{\gamma}_i^{k}$. Now, since the dual function in \eqref{eq:dualFunction} is convex, we can minimize it by performing a stochastic subgradient descent. Then, the dual updates are given by the following expressions
\begin{align}
\gamma_i^k[t+1]:=\biggl[\gamma_i^k[t]+a_i^k[t] - x_i^k[t]+\sum_{j \in \mathcal{N}_i}\bigl( r_{ji}^k[t]-r_{ij}^k[t]\bigr)\biggr]^+ 
\label{eq:dualQueue}
\end{align}
\begin{align}
\beta_i[t+1]:=\biggl[\beta_i[t] -e_i[t] + \sum_{k \in \mathcal{K}} \sum_{j \in \mathcal{N}_i} r_{ij}^k[t] \biggr]^+\label{eq:dualBattery}
\end{align}
where $[\cdot]^+$ is the projection on the nonnegative orthant. For compactness, we also express the dual updates in vector form as $\boldsymbol{\lambda}[t+1]:=\left[\boldsymbol{\lambda}[t]-\mathbf{s}[t]\right]^+$, where $\mathbf{s}[t]$ corresponds to the vector collecting the stochastic subgradients. Since the algorithm that we propose is designed to be run in an online fashion, we have considered a fixed step size in the dual updates. Specifically, we have used a unit step size. This allows a clear comparison between dual variables and data queues and battery dynamics as outlined in Figure \ref{fig:relQueuesBatteries}. For the case of the data queues, the difference between their dynamics \eqref{eq:queueEvol} and those of their Lagrange multiplier counterparts \eqref{eq:dualQueue} is given by the auxiliary variable in the dual update. Assume a packet is either routed or not, i.e., $r_{ij}^k[t] \in \{0,1\}$. Then the dual variables $q_i^k[t]$ follow the data queues $\gamma_i^k[t]$ until $\gamma_i^k[t] > \bar{\gamma}_i^{k}$, at which point, the dual variables are pushed back by the auxiliary variable $x_i^k [t]=\bar{x}_i^k$. From this point forward, the queue and multiplier dynamics lose their symmetry, coupling again when the queue empties. In a similar way, a comparison can also be drawn between the battery dynamics \eqref{eq:batteryEvol} and the battery dual update \eqref{eq:dualBattery}. In this case, the symmetry exists in a mirrored way, as the relationship between the battery state $b_i[t]$ and its multipliers $\beta_i[t]$ is given by $b_i[t]=b_i^{\max}-\beta_i[t]$. Different from the case of data queues, the coupling between the battery state and its multipliers is never lost.

\begin{figure}[t!]
    \centering
    \subfigure[Data queues.]
    {
        \includegraphics[scale=1.21]{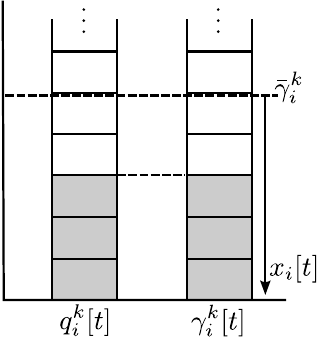}
        \label{fig:relQueues}
    }
    \subfigure[Batteries.]
    {
        \includegraphics[scale=1.21]{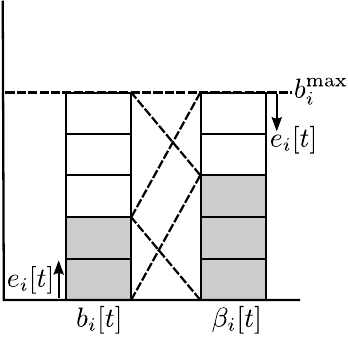}
        \label{fig:relBatteries}
    }
    \caption{Relationship between data queues, batteries and their Lagrange multipliers. Example with $\bar{\gamma}_i^{k}=5$ and $b_i^{\max}=5$. }
    \label{fig:relQueuesBatteries}
\end{figure}

Next, we consider some choices of the objective function $f_{ij}^k \bigl( r_{ij}^k \bigr)$ in the optimization problem \eqref{eq:optProblem} which lead to familiar formulations of the backpressure algorithm adapted to the energy harvesting process. The steps of the two resulting policies are summarized in Algorithm \ref{alg:Algorithm}.

\begin{algorithm}[t]
    \caption{Stochastic Backpressure with Energy Harvesting.}
    \label{alg:Algorithm}
    \begin{algorithmic}[1]
		\State \textbf{Initialize:} Set $\gamma_i^k[0]:=0$ and $\beta_i[0]:=b_i^{\max}-b_i[0]$.
        \State \textbf{Step 1(a):} Routing-scheduling decision (SBP-EH).
        \State $r^k_{ij}[t] := \argmax\limits_{\sum\limits_{k,j} r_{ij}^k \leq 1} 
 \sum\limits_{k \in \mathcal{K}} \sum\limits_{j \in \mathcal{N}_i} r_{ij}^k \bigl(  w_{ij}^k + \gamma_i^k[t] - \gamma_j^k[t] -\beta_i[t]\bigr)$
%        \State $k_{ij}^\star := \argmax_{k} \left(  w_{ij}^k + \gamma_i^k[t] - \gamma_j^k[t] -\beta_i[t] \right)$        
%        \State $r_{ij}^{k_{ij}^\star}[t]:=\mathbb{I} \left(  w_{ij}^k + \gamma_i^k[t] - \gamma_j^k[t] -\beta_i[t] > 0 \right)$                
        \State \textbf{Step 1(b):} Routing-scheduling decision (SSBP-EH).
        \State $r^k_{ij}[t] := \frac{1}{2}\biggl[w_{ij}^k + \gamma_i^k[t] - \gamma_j^k[t] -\beta_i[t] - \nu_i[t] \biggr]^+$        
        \State \textbf{Step 2:} Compute auxiliary variable.
        \State $x_i^k [t] := \argmax\limits_{x_i^k \in [0,\bar{x}_i^k]} x_i^k \left( \gamma_i^k[t] - \bar{\gamma}_i^{k}\right)$        
        \State \textbf{Step 3:} Update dual variables.
        \State $\gamma_i^k[t+1]:=\biggl[\gamma_i^k[t]+a_i^k[t]-x_i^k[t]+\sum\limits_{j \in \mathcal{N}_i} \bigl( r_{ji}^k[t]- r_{ij}^k[t]\bigr) \biggr]^+$  
        \State $\beta_i[t+1]:=\biggl[\beta_i[t] -e_i[t] + \sum\limits_{k \in \mathcal{K}} \sum\limits_{j \in \mathcal{N}_i} r_{ij}^k[t] \biggr]^+$
        \State \textbf{Step 4:} For all neighbors $j \in \mathcal{N}_i$, send dual variables $\gamma_i^k[t+1]$ and receive dual variables $\gamma_j^k[t+1]$.
        \State \textbf{Step 5:} Set $t:=t+1$ and go to Step 1.
        \end{algorithmic}
\end{algorithm}

\subsection{Stochastic Backpressure with Energy Harvesting (SBP-EH)}
\label{subsec:SBP-EH}

 Consider functions $f_{ij}^k \bigl( r_{ij}^k \bigr)$ which are linear with respect to the routing variables, i.e., taking the form $f_{ij}^k \bigl( r_{ij}^k \bigr)=w_{ij}^k r_{ij}^k$, where $w_{ij}^k$ is an arbitrary weight. In this case, we recover a version of the stochastic backpressure algorithm adapted to the energy harvesting process. For a linear objective function, the maximization in \eqref{eq:dualFunctionMax} leads to the routing variables
 \begin{align}
r^k_{ij}[t] := \argmax_{\sum\limits_{k,j} r_{ij}^k \leq 1} 
 \sum_{k \in \mathcal{K}} \sum_{j \in \mathcal{N}_i} r_{ij}^k \left(  w_{ij}^k + \gamma_i^k[t] - \gamma_j^k[t] -\beta_i[t]\right).
\label{eq:routingBP}
\end{align}
To solve the maximization in \eqref{eq:routingBP} it suffices to find the flow over the neighboring nodes with the largest differential $w_{ij}^k + \gamma_i^k[t] - \gamma_j^k[t] -\beta_i[t]$ and if it is positive, set its corresponding routing variable $r_{ij}^k[t]$ to one while the other variables are kept to zero. This algorithm, when $w_{ij}^k=0$, is analogous to the stochastic form of backpressure. In the classical backpressure algorithm, the flow with the largest queue differential $q_i^k[t] - q_j^k[t]$ is chosen. Interpreted in its stochastic form, the flow with the largest Lagrange multiplier difference $\gamma_i^k[t] - \gamma_j^k[t]$ is chosen. In the SBP-EH policy, the stochastic form of backpressure adds the battery multiplier $\beta_i[t]$. As the battery depletes, the value of $\beta_i[t]$ increases and the pressure to transmit of this node decreases.

\subsection{Stochastic Soft Backpressure with Energy Harvesting (SSBP-EH)}
\label{subsec:SSBP-EH}

Now, we consider a quadratic plus linear term function given by $f_{ij}^k \bigl( r_{ij}^k \bigr)= -\bigl(r_{ij}^k\bigr)^2+w_{ij}^k r_{ij}^k$. This leads to a stochastic soft backpressure algorithm \cite{ribeiro2009stochastic}, where the routing variables obtained by the maximization in \eqref{eq:dualFunctionMax} are given by
 \begin{align}
r^k_{ij}[t] := \frac{1}{2}\biggl[w_{ij}^k + \gamma_i^k[t] - \gamma_j^k[t] -\beta_i[t] - \nu_i[t] \biggr]^+,
\label{eq:routingSSBP}
\end{align}
where $\nu_i[t]$ are the Lagrange multipliers ensuring $\sum_{k \in \mathcal{K}} \sum_{j \in \mathcal{N}_i} r_{ij}^k \leq 1$ for all $i \in \mathcal{N}$. This expression can be understood a form of inverse waterfilling. An example of this solution is shown in Figure \ref{fig:inverseWF}. Let us construct rectangles of height $H_{ij}^k[t]=w_{ij}^k + \gamma_i^k[t] - \gamma_j^k[t] -\beta_i[t]$ and scale them by the widths $W_{ij}^k[t]=1/2$. For each node, every possible flow and neighbor routing destination is represented by one of these rectangles. Then, water is poured from the bottom, in an inverse manner until the  $\nu_i[t]$ waterlevel is reached. The resulting area of water filled inside the rectangles represents the probability mass function of the routing variables. Then, the node takes its routing decision by drawing a sample from this distribution.

While not as simple as the SBP-EH algorithm, the SSBP-EH algorithm presents an important improvement over the former. The introduction of an strongly concave objective function allows the dual function in \eqref{eq:dualFunction} to be differentiable. This, in turn, makes the algorithm take the form of an stochastic gradient rather than a stochastic subgradient (which is the case of SBP-EH), therefore improving the expected rate of stabilization of the algorithm from $O(1/\sqrt{t})$ to $O(1/t)$ \cite[Chapter 3.2]{bertsekas2015convex}.

\begin{figure}[!t]
	\centering
	\includegraphics[width=0.95\columnwidth ]{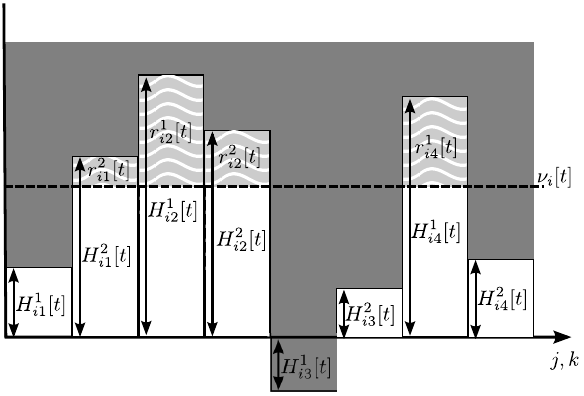}
	\caption{Example of a SSBP-EH routing decision for a node node serving two flows and four neighbors.}
	\label{fig:inverseWF}
\end{figure}

\section{Causality and Stability Analysis}
\label{sec:Stability}

In this section, we provide theoretical guarantees on the behavior of the proposed policies. On one hand, we establish the conditions under which the routing policies generated by Algorithm \ref{alg:Algorithm} satisfy the energy causality constraints \eqref{eq:energyCausality}. And, on the other hand, we provide stability guarantees on the network queues.

\subsection{Energy Causality}

As we mentioned previously in Section \ref{sec:SystemModel}, the presence of the energy harvesting constraints in the stochastic optimization problem \eqref{eq:optProblemOriginal} introduces the question of causality. In order to have a tractable problem, we have introduced the energy harvesting constraints in an average sense to the routing-scheduling problem. This includes an additional issue, as not all possible solutions satisfy the original causality constraints \eqref{eq:energyCausality}  for all time slots. In order to deal with this, we have modified the  problem formulation with the introduction of an auxiliary variable. By appropriately choosing the domain of this auxiliary variable and the nodes' battery capacity, we can ensure that the causality constraints are satisfied.

\begin{proposition}[Energy Causality]
\label{prop:EnergyCausality}
Let the data arrivals of a node $i \in \mathcal{N}$ satisfy $a_i^k[t] \leq \bar{a}_i^k$ for all $k \in \mathcal{K}$ flows and all time slots $t$. Further, let $\bar{x}_i^k \geq \bar{\gamma}_i^k +  \bar{a}_i^k + |\mathcal{N}_i|$ and let the battery capacity satisfy $b_i^{\max} \geq w_{ij}^k + \bar{\gamma}_i^k +  \bar{a}_i^k + |\mathcal{N}_i|$. Then both the SBP-EH and SSBP-EH policies shown in Algorithm \ref{alg:Algorithm} satisfy the energy consumption causality constraint $\sum_{k \in \mathcal{K}} \sum_{j \in \mathcal{N}_i} r_{ij}^k[t]  \leq b_i[t]$ of this node for all time slots.
\end{proposition}
\begin{proof}
To satisfy the energy causality constraints it suffices to show that no transmission occurs when there is no available energy in the battery. This is to say that $r_{ij}^k[t]=0$ for all $j,k$ if $b_i[t]=0$. In expressions \eqref{eq:routingBP} and \eqref{eq:routingSSBP}, corresponding to the SBP-EH and SSBP-EH algorithms, it suffices to ensure that $w_{ij}^k + \gamma_i^k[t] - \gamma_j^k[t] -\beta_i[t] < 0$ when the battery is empty. In this case, when $b_i[t]=0$, the battery dual update takes the value $\beta_i[t]=b_i^{\max}$. By the dual update \eqref{eq:dualQueue} and the minimum value of $\bar{x}_i^k$, the data arrival bound $\bar{a}_i^k$ and the number of neighbors $|\mathcal{N}_i|$, we can upper bound the multiplier difference by $\gamma_i^k[t] - \gamma_j^k[t] \leq \bar{\gamma}_i^k + \bar{a}_i^k + |\mathcal{N}_i| $ over all time slots $t$. We can write then $w_{ij}^k + \bar{\gamma}_i^k + \bar{a}_i^k + |\mathcal{N}_i| - b_i^{\max} \leq 0$, and since $b_i^{\max} \geq w_{ij}^k + \bar{\gamma}_i^k + \bar{a}_i^k + |\mathcal{N}_i|$, this ensures that $r_{ij}^k[t]=0$. Hence, satisfying the energy causality constraint $\sum_{k \in \mathcal{K}} \sum_{j \in \mathcal{N}_i} r_{ij}^k[t]  \leq b_i[t]$ for all time slots.
\end{proof}

In order to ensure that the energy causality constraints are satisfied, the stochastic subgradients are required to be bounded. This, in turn, forces the probability distribution of the data arrival process to be bounded above by a constant $\bar{a}_i^k$. In practice, for the case in which the probability distribution is not bounded, when a time slot with over $\bar{a}_i^k$ packets occurs, only $\bar{a}_i^k$ data packets can be kept in the queue and the rest must be discarded to satisfy the energy causality constraints.

\begin{remark}\label{rmk_1}
It follows from Proposition \ref{prop:EnergyCausality} that the consumption of energy is causal. Namely, that energy is not consumed before it is harvested. Although we have not remarked it here, the same is true of packet transmissions. Namely, packets are not transmitted before they arrive. The causality of queues is a natural property of stochastic dual descent but the causality of energy is not. The latter requires the introduction of the auxiliary variables $x_i^k[t]$ to the problem formulation. This modification is such that the instantaneous values of the dual variables generated by Algorithm \ref{alg:Algorithm} can be shown to be bounded above. In turn, having bounded dual variables allows us to establish conditions on $\bar{x}_i^k$ that ensure that the energy causality constraints \eqref{eq:energyCausality} are satisfied for all time slots as shown in Proposition \ref{prop:EnergyCausality}. The dual variables associated with energy constraints that would be generated by Algorithm \ref{alg:Algorithm} could become arbitrarily large if we write a stochastic dual descent algorithm for the original problem in \eqref{eq:optProblemOriginal}. The addition of the variables $\bar{x}_i^k$ enforces causality but it is not without a cost; see Remark \ref{rmk_2}.
\end{remark}

\subsection{Queue Stability}

Now, we provide guarantees on the queue stability of the proposed policies. Different from other works (such as\cite{gatzianas2010control,huang2013utility}), which analyze queue stability with Lyapunov drift notions, we resort to duality theory arguments. We do this by leveraging on the fact that the proposed algorithm is a type of stochastic subgradient algorithm. The approach we take to showing that our algorithm makes the queues stable in the sense of \eqref{eq:queueStability} is to show that the solution provided by Algorithm \ref{alg:Algorithm} satisfies the queue stability constraints \eqref{eq:optProblemConstraintQueue} almost surely. Then, we show that if the optimal queue multipliers $\gamma_i^k$ are upper bounded by $\bar{\gamma}_i^{k}$, the solution provided by Algorithm \ref{alg:Algorithm} also satisfies the stability constraint without auxiliary variable \eqref{eq:optProblemConstraintQueueOriginal}. Hence, the data queues satisfy the stability condition \eqref{eq:queueStability}.

First, we start by recalling a common property of the stochastic subgradient.

\begin{proposition}
\label{prop:avgSubgradient}
Given the dual variables $\boldsymbol{\lambda}[t]$, the conditional expected value $\mbE\left[\mathbf{s}[t] | \boldsymbol{\lambda}[t]\right]$ of the stochastic subgradient $\mathbf{s}[t]$ is a subgradient of the dual function. Namely, for any $\boldsymbol{\lambda}$,
\begin{align}
\mbE\left[ \mathbf{s}^T[t] | \boldsymbol{\lambda}[t]\right]\left(\boldsymbol{\lambda}[t] - \boldsymbol{\lambda}\right)
\geq g(\boldsymbol{\lambda}[t]) -g(\boldsymbol{\lambda}).
\label{prop:avgSubgradientEQ}
\end{align}
\end{proposition}
\begin{proof}
Take the Lagrangian \eqref{eq:Lagrangian} and substitute the ergodic definitions $\mbE \bigl[ a^k_i[t] \bigr] = a^k_i$ and $\mbE \bigl[e_i[t] \bigr] = e_i$. Then, the resulting Lagrangian is given by
\begin{align}
\mathcal{L}(\mathbf{r},\boldsymbol{\lambda})&=\sum_{i \in \mathcal{N}} \sum_{k \in \mathcal{K}}  \sum_{j \in \mathcal{N}_i} f_{ij}^k \left( r_{ij}^k \right) 
	- \sum_{i \in \mathcal{N}} \sum_{k \in \mathcal{K}} \bar{\gamma}_i^{k} x_i^k \nonumber\\
	&+\sum_{k \in \mathcal{K}}   \sum_{i \in \mathcal{N}} \gamma_i^k\biggl(\sum_{j \in \mathcal{N}_i} \bigl(r_{ij}^k - r_{ji}^k\bigr) + x_i^k  -\mbE\bigl[ a^k_i[t] \bigr] \biggr) \nonumber\\
	&+\sum_{i \in \mathcal{N}} \beta_i\biggl(\mbE \bigl[e_i[t] \bigr] - \sum_{k \in \mathcal{K}} \sum_{j \in \mathcal{N}_i} r_{ij}^k \biggr).
\end{align}
Now, recall that the dual function is then given by $g(\boldsymbol{\lambda})=\max_{\mathbf{r} \in \mathcal{R}} \mathcal{L}(\mathbf{r},\boldsymbol{\lambda})$, and consider the dual function at time $t$, given by $g(\boldsymbol{\lambda}[t])$. The primal maximization of this dual function is given by the variables $r_{ij}^k[t]$ and $x_i^k[t]$ in \eqref{eq:dualFunctionMax} and \eqref{eq:dualFunctionMaxAuxiliary}, respectively. Hence, we can write the dual function as
\begin{align}
g(\boldsymbol{\lambda}&[t])=\sum_{i \in \mathcal{N}} \sum_{k \in \mathcal{K}}  \sum_{j \in \mathcal{N}_i} f_{ij}^k \left( r_{ij}^k[t] \right) 
	- \sum_{i \in \mathcal{N}} \sum_{k \in \mathcal{K}} \bar{\gamma}_i^{k} x_i^k[t] \nonumber\\
	&+\sum_{k \in \mathcal{K}}   \sum_{i \in \mathcal{N}} \gamma_i^k[t]\mbE\biggl[\sum_{j \in \mathcal{N}_i} \bigl(r_{ij}^k[t] - r_{ji}^k[t]\bigr) + x_i^k[t] - a^k_i[t] \biggr] \nonumber\\
	&+\sum_{i \in \mathcal{N}} \beta_i[t] \mbE\biggl[ e_i[t] - \sum_{k \in \mathcal{K}} \sum_{j \in \mathcal{N}_i} r_{ij}^k[t] \biggr],
\end{align}
where we have moved the expectation operator $\mbE[\cdot]$ out of the subgradients due to its linearity. Then we can use the compact notation for the multiplier vector $\boldsymbol{\lambda}[t]$ and the subgradient $\mathbf{s}[t]$, and substitute the conditional expected value of the subgradients $\mbE\left[\mathbf{s}[t] | \boldsymbol{\lambda}[t]\right]$ to obtain
\begin{align}
g(\boldsymbol{\lambda}[t])&= \sum_{i \in \mathcal{N}} \sum_{k \in \mathcal{K}}  \sum_{j \in \mathcal{N}_i} f_{ij}^k \left( r_{ij}^k[t] \right) - \sum_{i \in \mathcal{N}} \sum_{k \in \mathcal{K}} \bar{\gamma}_i^{k} x_i^k[t]\nonumber\\
&+\mbE\left[ \mathbf{s}^T[t] | \boldsymbol{\lambda}[t]\right]\boldsymbol{\lambda}[t].
\label{prop:avgSubgradientDualLT}
\end{align}
For any arbitrary $\boldsymbol{\lambda}$ we simply have
\begin{align}
g(\boldsymbol{\lambda}) &\geq \sum_{i \in \mathcal{N}} \sum_{k \in \mathcal{K}}  \sum_{j \in \mathcal{N}_i} f_{ij}^k \left( r_{ij}^k[t] \right) - \sum_{i \in \mathcal{N}} \sum_{k \in \mathcal{K}} \bar{\gamma}_i^{k} x_i^k[t]\nonumber\\
&+\mbE\left[ \mathbf{s}^T[t] | \boldsymbol{\lambda}[t]\right]\boldsymbol{\lambda}
\label{prop:avgSubgradientDualL}.
\end{align}
Then it simply suffices to subtract expression \eqref{prop:avgSubgradientDualL} from \eqref{prop:avgSubgradientDualLT} to obtain inequality \eqref{prop:avgSubgradientEQ}.
\end{proof}

Proposition \ref{prop:avgSubgradient} shows that the stochastic subgradient is an average descent direction of the dual function $g(\boldsymbol{\lambda}[t])$. Now, we proceed to quantify the average descent distance of the dual update.
\begin{lemma}
\label{lmm:avgDescentDual}
Consider the dual updates of Algorithm \ref{alg:Algorithm} given by \eqref{eq:dualQueue} and \eqref{eq:dualBattery}, and let $\mbE\left[\|\mathbf{s}[t] \|^2| \boldsymbol{\lambda}[t]\right] \leq S^2$ be a bound on the second moment of the norm of the stochastic subgradients $\mathbf{s}[t]$. Then, the dual updates satisfy the inequality
\begin{align}
\label{lmm:avgDescentDualeq1}
\mbE\left[\|\boldsymbol{\lambda}[t+1]-\boldsymbol{\lambda}^{\star}\|^2 | \boldsymbol{\lambda}[t]\right]
\leq &\|\boldsymbol{\lambda}[t] - \boldsymbol{\lambda}^{\star}\|^2 + S^2 \nonumber\\
&-2 \left(g(\boldsymbol{\lambda}[t]) -g(\boldsymbol{\lambda}^{\star}) \right)
\end{align}
\end{lemma}
\begin{proof}
Start by considering the squared distance between the dual variables at time $t+1$ and their optimal value. This distance is given by $\|\boldsymbol{\lambda}[t+1]-\boldsymbol{\lambda}^{\star}\|^2$. Then, we substitute the dual variable $\boldsymbol{\lambda}[t+1]$ by its update $\boldsymbol{\lambda}[t+1]=\left[\boldsymbol{\lambda}[t]-\mathbf{s}[t]\right]^+$. Then, since the projection is nonexpansive we can upper bound the aforementioned distance by
\begin{align}
\|\boldsymbol{\lambda}[t+1]-\boldsymbol{\lambda}^{\star}\|^2 \leq 
\|\boldsymbol{\lambda}[t]-\mathbf{s}[t] - \boldsymbol{\lambda}^{\star} \|^2.
\end{align}
Then, we simply expand the square norm to obtain the expression
\begin{align}
\|\boldsymbol{\lambda}[t+1]-\boldsymbol{\lambda}^{\star}\|^2 \leq &
\|\boldsymbol{\lambda}[t]-\boldsymbol{\lambda}^{\star}\|^2 + \| \mathbf{s}[t]\|^2 \nonumber\\
&-2\mathbf{s}^T[t]\left(\boldsymbol{\lambda}[t] - \boldsymbol{\lambda}^{\star}\right).
\end{align}
Now, by taking the expectation conditioned by $\boldsymbol{\lambda}[t]$ on both sides we obtain
\begin{align}
\mbE\bigl[\|\boldsymbol{\lambda}[t+1]-\boldsymbol{\lambda}^{\star}\|^2 | \boldsymbol{\lambda}[t]&\bigr]  \leq 
\|\boldsymbol{\lambda}[t]-\boldsymbol{\lambda}^{\star}\|^2 +
\mbE\left[\|\mathbf{s}[t] \|^2 | \boldsymbol{\lambda}[t]\right] \nonumber\\
&-2\mbE \left[ \mathbf{s}^T[t] | \boldsymbol{\lambda}[t] \right] \left(\boldsymbol{\lambda}[t] - \boldsymbol{\lambda}^{\star}\right)
\end{align}
And then by substituting the second term on the right hand side by the bound $\mbE\left[\|\mathbf{s}[t] \|^2| \boldsymbol{\lambda}[t]\right] \leq S^2$ and the third term by the application of Proposition \ref{prop:avgSubgradient} with $\boldsymbol{\lambda} = \boldsymbol{\lambda}^{\star}$, we have expression \eqref{lmm:avgDescentDualeq1}.
\end{proof}

Then, we leverage on this lemma to show that Algorithm \ref{alg:Algorithm} converges to a neighborhood of the optimal solution of the dual function.

\begin{lemma}
\label{lmm:DualGap}
Consider the dual updates of Algorithm \ref{alg:Algorithm} given by \eqref{eq:dualQueue} and \eqref{eq:dualBattery}, and let $\mbE\left[\|\mathbf{s}[t] \|^2| \boldsymbol{\lambda}[t]\right] \leq S^2$ be a bound on the second moment of the norm of the stochastic subgradients $\mathbf{s}[t]$. Assume that the dual variable $\boldsymbol{\lambda}[T]$ is given for an arbitrary time $T$ and define as $\boldsymbol{\lambda}_{\mathrm{best}}[t] := \argmin_{\boldsymbol{\lambda}[l]} g(\boldsymbol{\lambda}[l])$ the dual variable leading to the best value of the of the dual function for the interval $l \in [T,t]$. Then, we have
\begin{align}
\label{eq:lemma}
\lim\limits_{t \to \infty} g(\boldsymbol{\lambda}_{\mathrm{best}}[t]|\boldsymbol{\lambda}[T]) \leq g(\boldsymbol{\lambda}^{\star}) + \frac{S^2}{2} \quad \text{a.s.}
\end{align}
\end{lemma}
\begin{proof}
For ease of exposition, let $T=0$. Then, define the stopped process $\alpha[t]$, tracking the distance between the dual variables at time $t$ and their optimal value, i.e., $\|\boldsymbol{\lambda}[t]-\boldsymbol{\lambda}^{\star}\|^2 $, until the optimality gap $g(\boldsymbol{\lambda}[t]) -g(\boldsymbol{\lambda}^{\star})$ falls below $S^2/2$. This expression is given by
\begin{align}
\label{eq:alpha}
\alpha [t] :=\|\boldsymbol{\lambda}[t]-\boldsymbol{\lambda}^{\star}\|^2
 \mathbb{I}\left\{ g(\boldsymbol{\lambda}_{\mathrm{best}}[t]) -g(\boldsymbol{\lambda}^{\star}) > S^2/2 \right\}.
\end{align}
where $\mathbb{I}\{\cdot\}$ denotes the indicator function. In a similar way, define the sequence $\beta [t]$ which follows $2\left( g(\boldsymbol{\lambda}[t]) -g(\boldsymbol{\lambda}^{\star})  \right)-S^2$ until the optimality gap $g(\boldsymbol{\lambda}[t]) -g(\boldsymbol{\lambda}^{\star})$ becomes smaller than $S^2/2$,
\begin{align}
\label{eq:beta}
\beta [t] :=\bigl(2\bigl(g(\boldsymbol{\lambda}[t]) -&g(\boldsymbol{\lambda}^{\star})  \bigr)-S^2\bigr) \nonumber\\
& \mathbb{I}\left\{ g(\boldsymbol{\lambda}_{\mathrm{best}}[t]) -g(\boldsymbol{\lambda}^{\star}) > S^2/2 \right\}.
\end{align}
Now, let $\mathcal{F}[t]$ be the filtration measuring $\alpha[t]$ and $\beta[t]$. Since $\alpha[t]$ and $\beta[t]$ are completely determined by $\boldsymbol{\lambda}[t]$, and $\boldsymbol{\lambda}[t]$ is a Markov process, conditioning on $\mathcal{F}[t]$ is equivalent to conditioning on $\boldsymbol{\lambda}[t]$. Hence, by application of Lemma \ref{lmm:avgDescentDual}, we can write $\mbE\left[ \alpha[t+1] | \mathcal{F}[t] \right] \leq \alpha[t] -\beta[t]$. Since by definitions \eqref{eq:alpha} and \eqref{eq:beta}, the processes $\alpha [t]$ and $\beta [t]$ are nonnegative, the sequence $\alpha [t]$ follows a  supermartingale expression. Then, by the supermartingale convergence theorem \cite[Theorem 5.2.9]{durrett2010probability}, the sequence $\alpha [t]$ converges almost surely, and the sum $\sum_{t=1}^{\infty}\beta[t] < \infty$ is almost surely finite. The latter implies that $\lim\inf_{t \to \infty} \beta[t] = 0$ almost surely. Given the definition of $\beta [t]$, this is implied by either of two events. (i) If the indicator function goes to zero, i.e., $g(\boldsymbol{\lambda}_{\mathrm{best}}[t]) -g(\boldsymbol{\lambda}^{\star}) \leq S^2/2$ for a large $t$; or (ii) $\lim\inf_{t \to \infty}2\bigl(g(\boldsymbol{\lambda}[t]) -g(\boldsymbol{\lambda}^{\star})  \bigr)-S^2 = 0$. From any of those events, expression \eqref{eq:lemma} follows.
\end{proof}

The convergence of the dual function as asserted in the previous lemma allows us to prove that the sequences of routing decision $\{r_{ij}^k[t]\}_{t=1}^{\infty}$ and auxiliary variables $\{x_i^k[t]\}_{t=1}^{\infty}$ generated by Algorithm \ref{alg:Algorithm}  are almost surely feasible.

\begin{proposition}[Auxiliary Feasibility]
\label{prop:AuxFeasibility}
Assume there exist strictly feasible primal variables $r_{ij}^k$ and $x_i^k$ such that $\sum_{j \in \mathcal{N}_i} r_{ij}^k - \sum_{j \in \mathcal{N}_i} r_{ji}^k-a^k_i + x_i^k > \xi$ and $e_i - \sum_{k \in \mathcal{K}} \sum_{j \in \mathcal{N}_i} r_{ij}^k > \xi$, for some $\xi > 0$. Then, the constraints \eqref{eq:optProblemConstraintQueue} and \eqref{eq:optProblemConstraintEnergy} are almost surely satisfied by Algorithm \ref{alg:Algorithm}.
\end{proposition}
\begin{proof}
First, let us collect the feasible routing variables $r_{ij}^k$ and auxiliary variables $x_i^k$ in the vector $\hat{\mathbf{r}}=\{r_{ij}^k, x_i^k\}$. Then, if there exist strictly feasible variables $\hat{\mathbf{r}}$, we can bound the value of the dual function $g(\boldsymbol{\lambda})$ as follows. The dual function is defined as the maximum over primal variables $g(\boldsymbol{\lambda})=\max_{\mathbf{r}} \mathcal{L}(\mathbf{r},\boldsymbol{\lambda})$, hence $g(\boldsymbol{\lambda}) \geq \mathcal{L}(\hat{\mathbf{r}},\boldsymbol{\lambda})$. From this, by using the $\sum_{j \in \mathcal{N}_i} r_{ij}^k - \sum_{j \in \mathcal{N}_i} r_{ji}^k-a^k_i + x_i^k > \xi$ and $e_i - \sum_{k \in \mathcal{K}} \sum_{j \in \mathcal{N}_i} r_{ij}^k > \xi$ terms we establish the following bound
\begin{align} 
g(\boldsymbol{\lambda}) \geq 
	\sum_{i \in \mathcal{N}} \sum_{k \in \mathcal{K}}  \sum_{j \in \mathcal{N}_i}   f_{ij}^k \left( r_{ij}^k \right)  - \sum_{i \in \mathcal{N}} \sum_{k \in \mathcal{K}} \bar{\gamma}_i^{k} x_i^k +
	\xi \boldsymbol{\lambda}^T\mathbf{1}.
\end{align}
Then, by simply reordering terms we obtain the following upper bound on the dual variables
\begin{align} 
\boldsymbol{\lambda} \leq \frac{1}{\xi} \biggl(g(\boldsymbol{\lambda})-\sum_{i \in \mathcal{N}} \sum_{k \in \mathcal{K}} \bigl( \sum_{j \in \mathcal{N}_i}   f_{ij}^k \left( r_{ij}^k \right)  + \bar{\gamma}_i^{k} x_i^k \bigr) \biggr).
\end{align}
Lemma \ref{lmm:DualGap} certifies the existence of a time $t \geq T_0$ for which $g(\boldsymbol{\lambda}[t]) \leq g(\boldsymbol{\lambda}^{\star}) + S^2/2$. Hence,
\begin{align} 
\boldsymbol{\lambda}[t] \leq \frac{1}{\xi} \biggl( g(\boldsymbol{\lambda}^{\star}) + \frac{S^2}{2} -\sum_{i \in \mathcal{N}} \sum_{k \in \mathcal{K}}  \bigl( \sum_{j \in \mathcal{N}_i}   f_{ij}^k \left( r_{ij}^k \right)  + \bar{\gamma}_i^{k} x_i^k \bigr) \biggr)
	\label{eq:lambdaUpperBound} 
\end{align}
for $t \geq T_0$. Now, recall that the feasibility conditions \eqref{eq:optProblemConstraintQueue} and \eqref{eq:optProblemConstraintEnergy} are given by the limits
\begin{align} 
\lim_{t \to \infty}\frac{1}{t} \sum\limits_{l=1}^{t}
\biggl(
\sum_{j \in \mathcal{N}_i} \bigl( r_{ij}^k[l] - r_{ji}^k[l] \bigr)-a^k_i[l]+x_i^k[t]
\biggr)\geq 0
\label{eq:feasQueue} \\
\lim_{t \to \infty}\frac{1}{t} \sum\limits_{l=1}^{t}
\biggl(
e_i[l] - \sum_{k \in \mathcal{K}} \sum_{j \in \mathcal{N}_i} r_{ij}^k[l] \biggr) \geq 0
\label{eq:feasBattery}
\end{align}
which, by recalling that the constraints are simply the stochastic subgradients $\mathbf{s}[t]$ of the problem, they can also be written in compact form as $\lim_{t \to \infty}\frac{1}{t} \sum_{l=1}^{t}\mathbf{s}[l] \geq 0$. Now, consider the dual updates \eqref{eq:dualQueue} and \eqref{eq:dualBattery} given by $\boldsymbol{\lambda}[t+1]=\left[\boldsymbol{\lambda}[t]-\mathbf{s}[t]\right]^+$. Since the $[\cdot]^+$ operator corresponds to a nonnegative projection, the dual variables can be lower bounded by removing the projection and recursively substituting the updates
\begin{align} 
\boldsymbol{\lambda}[t+1]\geq \boldsymbol{\lambda}[t]-\mathbf{s}[t] \geq \boldsymbol{\lambda}[1] -\sum_{l=1}^{t}\mathbf{s}[l] \geq - \sum_{l=1}^{t}\mathbf{s}[l].
\label{eq:lambdaLowerBound}
\end{align}
To prove almost sure feasibility, we will follow by contradiction. First, assume that conditions \eqref{eq:feasQueue} and  \eqref{eq:feasBattery} are infeasible. In compact form, this means the existence of a time $t \geq T_1$, for which there is a constant $\delta > 0$  such that $\frac{1}{t} \sum_{l=1}^{t}\mathbf{s}[l] \leq -\delta$. By substituting in \eqref{eq:lambdaLowerBound}, we have that the dual variables are lower bounded by $\boldsymbol{\lambda}[t+1]\geq \delta t$. Now, we can freely choose a time $t \geq T_2$ such that
\begin{align} 
\boldsymbol{\lambda}[t] > \frac{1}{\xi} \biggl( g(\boldsymbol{\lambda}^{\star}) + \frac{S^2}{2} -\sum_{i \in \mathcal{N}} \sum_{k \in \mathcal{K}} \bigl( \sum_{j \in \mathcal{N}_i}   f_{ij}^k \left( r_{ij}^k \right)  + \bar{\gamma}_i^{k}  x_i^k \bigr)\biggr)
\end{align}
for all $t \geq T_2$. However, this contradicts the upper bound established in \eqref{eq:lambdaUpperBound}. This means that there do not exist sequences generated by Algorithm \ref{alg:Algorithm} such that \eqref{eq:feasQueue} and \eqref{eq:feasBattery} are not satisfied. Therefore, the constraints \eqref{eq:optProblemConstraintQueue} and \eqref{eq:optProblemConstraintEnergy} are satisfied almost surely. 
\end{proof}

Finally, it suffices to show that if the optimal dual variables are upper bounded by the constants $\bar{\gamma}_i^{k}$, the system satisfies the original problem without the auxiliary variable. Thus, satisfying the original constraint \eqref{eq:optProblemConstraintQueueOriginal} and hence the queue stability condition \eqref{eq:queueStability}.

\begin{proposition}[Feasibility]
\label{prop:Feasibility}
Assume that there exist strictly feasible routing variables $r_{ij}^k$ such that $\sum_{j \in \mathcal{N}_i} r_{ij}^k - \sum_{j \in \mathcal{N}_i} r_{ji}^k-a^k_i > \xi$ and $e_i - \sum_{k \in \mathcal{K}} \sum_{j \in \mathcal{N}_i} r_{ij}^k > \xi$, for some $\xi > 0$. Furthermore, assume the optimal Lagrange multipliers of the queue stability constraints satisfy $\gamma_i^{k,\star} \leq \bar{\gamma}_i^{k}$. Then, the constraints \eqref{eq:optProblemConstraintQueueOriginal} and \eqref{eq:optProblemConstraintEnergyOriginal} are almost surely satisfied by Algorithm \ref{alg:Algorithm}.
\end{proposition}
\begin{proof}
Take the difference between the Lagrangian \eqref{eq:Lagrangian} of the optimization problem with the auxiliary variable \eqref{eq:optProblem} and the original problem \eqref{eq:optProblemOriginal}. The difference between them is given by 
\begin{align}
\mathcal{L}(\mathbf{r},\boldsymbol{\lambda}) - \hat{\mathcal{L}}(\mathbf{r},\boldsymbol{\lambda})&=
\sum_{i \in \mathcal{N}} \sum_{k \in \mathcal{K}} \bigl(- \bar{\gamma}_i^{k} + \gamma_i^k + \theta_i^k - \nu_i^k \bigr) x_i^k \nonumber\\
&+\sum_{i \in \mathcal{N}} \sum_{k \in \mathcal{K}} \nu_i^k \bar{x}_i^k 
\label{eq:LagrangianDiff}
\end{align}
where $\theta_i^k \geq 0$ and $\nu_i^k \geq 0$ are the Lagrange multipliers of the $x_i^k \geq 0$ and $x_i^k \leq \bar{x}_i^k$ constraints, respectively. In order for both problems to be equivalent, the minimization of \eqref{eq:LagrangianDiff}, which is the solution of the dual problem, must be zero. This implies the existence of Lagrange multipliers satisfying the constraints $\bar{\gamma}_i^{k} - \gamma_i^k - \theta_i^k + \nu_i^k =0$, for all $i \in \mathcal{N}$ and $k \in \mathcal{K}$. Since $\gamma_i^{k,\star} \leq \bar{\gamma}_i^{k}$, the constraints can be satisfied by letting $\nu_i^{k,\star} = 0$, and $\theta_i^{k,\star}$ acting as a slack variable. Then, $\mathcal{L}(\mathbf{r},\boldsymbol{\lambda}) - \hat{\mathcal{L}}(\mathbf{r},\boldsymbol{\lambda})=0$, which implies that the optimal solution of both problems is the same. Since $\lim_{t \to \infty}\frac{1}{t} \sum_{l=1}^{t}x_i^k[l]=x_i^k$ \cite{ribeiro2010ergodic} and $x_i^k=0$, by Proposition \ref{prop:AuxFeasibility} the routing variables $r_{ij}^k$ of Algorithm \ref{alg:Algorithm} satisfy the constraint $\sum_{j \in \mathcal{N}_i} r_{ij}^k - \sum_{j \in \mathcal{N}_i} r_{ji}^k-a^k_i \geq 0$. 
\end{proof}

\begin{corollary}[Queue Stability]
\label{cor:QueueStability}
Consider the conditions of Proposition \ref{prop:Feasibility}.  Then, there exists a constant $Q$ such that for some arbitrary time $T$ the queues of the system under Algorithm \ref{alg:Algorithm} satisfy $\Pr{\max_{t \geq T} \|\mathbf{q}[t]\| \leq Q | \mathbf{q}[T]}=1$.
\end{corollary}
\begin{proof}
Denote by $\mathcal{F}_i^k[t]$ the filtration measuring $q_i^k[l]$. Then, since the routing variables $r_{ij}^k$ generated by Algorithm \ref{alg:Algorithm} satisfy $\sum_{j \in \mathcal{N}_i} r_{ij}^k - \sum_{j \in \mathcal{N}_i} r_{ji}^k-a^k_i \geq 0$, the queue evolution \eqref{eq:queueEvol} obeys the supermartingale expression $\mbE \left[q_i^k[t+1] | \mathcal{F}_i^k[t]\right] \leq q_i^k[t]$. By the supermartingale convergence theorem \cite[Theorem 5.2.9]{durrett2010probability}, the sequence $q_i^k[t]$ converges almost surely, therefore satisfying the stability condition $\Pr { \max_{t \geq T} \|\mathbf{q}[t]\| \leq Q | \mathbf{q}[T]}=1$.
\end{proof}

Given an appropriate choice of $\bar{\gamma}_i^{k}$ and feasible data and energy arrivals, Proposition \ref{prop:Feasibility} guarantees that the nodes route in average as many packets as they receive from neighbors and the arrival process (i.e., the constraint \eqref{eq:optProblemConstraintQueueOriginal} is satisfied). Then, Corollary \ref{cor:QueueStability} shows that this implies that the queues themselves are almost surely stable.

\begin{remark} \label{rmk_2}
As discussed in Remark \ref{rmk_1} the variables $x_i^k[t]$ are used to enforce causality of energy consumption. Doing this requires preventing transmissions when the battery is empty, something that is challenging because transmission decisions are largely based on the multipliers associated with queue stability constraints [cf. \eqref{eq:dualFunctionMax}]. The role of the variables $x_i^k[t]$ is to set $\gamma_i^k[t] = 0$ whenever they go over the threshold $\bar{\gamma}_i^{k}$, hence allowing us to establish the causality result in Proposition \ref{prop:EnergyCausality}. In a sense, this fools the agent into thinking that its queues are empty, thereby preventing transmissions when the node is out of energy. The drawback is that neighboring agents also see the queues as empty and decide to route traffic through agent $i$ when agent $i$ does not have the ability to handle such traffic. The convergence analysis implies that this short term drawback is not a problem in the long term because queues eventually stabilize with probability 1. Intuitively, this happens because the resetting of queues is rare, something that we indeed observe in numerical analyses; see Figure \ref{fig:QueuedPackets}.
\end{remark}

\begin{remark}
To deal with the causality problem we have resorted to the auxiliary formulation shown in \eqref{eq:optProblem}. This modified formulation can be seen as a way to bound the values of the dual variables. Hence, one might think that a more straightforward approach would be to simply project the dual variables to a restricted domain and remove the auxiliary variables $x_i^k[t]$ altogether from the optimization problem. We studied the consequences of this formulation in our previous work \cite{calvo2017stochastic}. The main drawback of using a dual projection is that the maximum value of the projection $\boldsymbol{\lambda}^{\max} $ must be lower bounded by
\begin{align}
\boldsymbol{\lambda}^{\max} >
\frac{1}{\xi}
\left( g(\boldsymbol{\lambda}^{\star}) + \frac{S^2}{2} -\sum_{i \in \mathcal{N}} \sum_{k \in \mathcal{K}}  \sum_{j \in \mathcal{N}_i}   f_{ij}^k \left( r_{ij}^k \right) \right).
\end{align}
This is a very loose requirement, which does not provide the clear theoretical guarantees that our proposed formulation offers. Nonetheless, the numerical results provided in \cite{calvo2017stochastic} seem to indicate that simply having the optimal dual variables in the range $[0,\boldsymbol{\lambda}^{\max}]$ can be sufficient. Thus, the simpler use of a projection of the dual variables can also be used as a less theoretically robust but practical option.
\end{remark}

\section{Numerical Results}
\label{sec:NumericalResults}

\begin{figure}[t]
    \centering
    \includegraphics[scale=1]{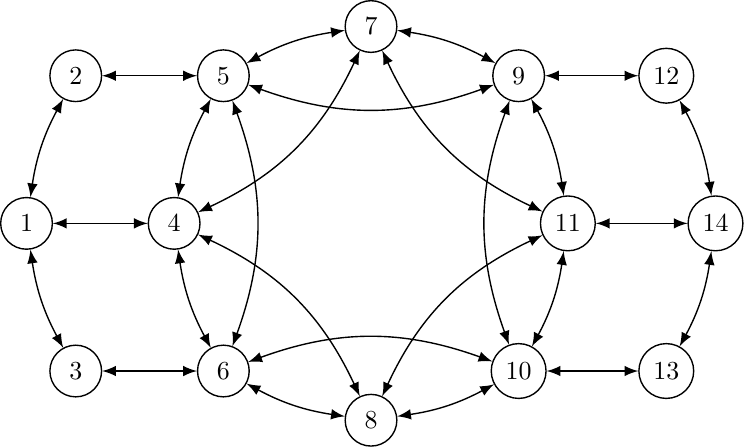}
    \caption{Connectivity graph of the simulated network.}
    \label{fig:Network}
\end{figure}

In this section, we conduct numerical experiments aimed at evaluating the performance of the proposed SBP-EH and SSBP-EH policies. As a means of comparison, when indicated, we also provide the non-energy harvesting counterparts of our proposed policies. Namely, the Stochastic Backpressure (SBP) \cite{tassiulas1992stability} and Stochastic Soft Backpressure (SSBP) \cite{ribeiro2009stochastic} policies. These policies correspond to solving \eqref{eq:optProblemNoEH}, the original optimization problem without the energy harvesting constraints, with the objective functions shown in Sections \ref{subsec:SBP-EH} and \ref{subsec:SSBP-EH}, respectively. Hence, these policies assume the availability of an unlimited energy supply. We consider the communication network shown in Figure \ref{fig:Network}, where we let nodes 1 and 14 act as sink nodes and the rest of the nodes support a single flow with packet arrival rates of $a_i^k=0.35$ packets per time slot. Moreover, we consider the nodes to be harvesting energy at a rate of $e_i=1$ units of energy per time slot and storing it in a battery of capacity $b_i^{\max}=15$. Furthermore, we set the routing weights to $w_{ij}^k=0$, and let $\bar{\gamma}_i^{k} = 10$.

\begin{figure}[t]
    \includegraphics[width=\columnwidth]{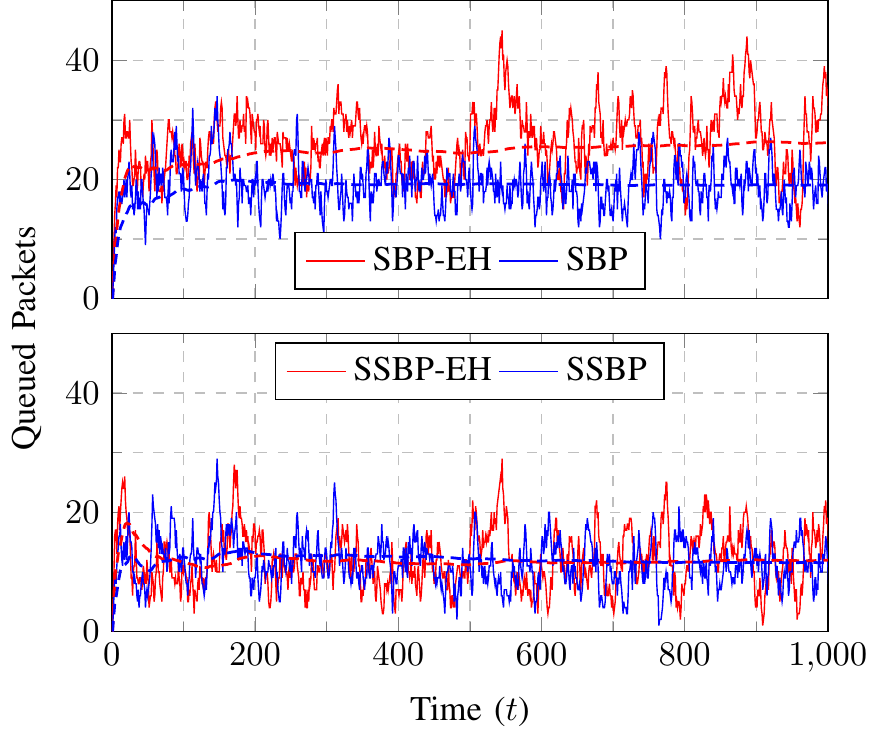}
	\caption{Total amount of packets queued in the network at each time slot. Average values are shown in dashed lines.}
\label{fig:QueuedPackets}
\end{figure}

\subsection{Network Queues}

First, we plot in Figure \ref{fig:QueuedPackets} a sample path of the total number of queued packets in the network as a function of the elapsed time. As expected, all the policies are capable of stabilizing the queues in the network. Due to the random nature of the processes, it is difficult to say exactly at which point stabilization occurs. Nonetheless, for the SBP and SBP-EH policies, the data queues seem to stop growing after around $t=200$ time slots. In the case of the SSBP and SSBP-EH policies, stabilization occurs much more rapidly rapidly, with less than $t=100$ time slots necessary to obtain stability. Also, both soft policies (SSBP and SSBP-EH) stabilize the queues with a lower number of average queued packets than their counterpart non-soft policies (SBP and SBP-EH). Namely, at $t=1000$, the average queued packets are $19.08$ for SBP and $26.22$ for SBP-EH. In the case of the soft policies, these numbers are much smaller, with $11.55$ and $11.97$ packets for SSBP and SSBP-EH, respectively. This also shows that the gap between the SSBP and SSBP-EH policies seems to vanish asymptotically ($3.63\%$ at $t=1000$), while this is not the case for the non-soft policies (a gap of $37.42\%$ at $t=1000$). This occurs due to the fact that the SBP and SBP-EH policies choose their routing policy by maximizing the difference between queue multipliers. Hence, making the decision indifferent to the actual value of the multipliers as long as their differences stay the same. For the SSBP and SSBP-EH policies, this situation does not occur due to their randomized nature. Hence, pushing for lower average queued packets. Furthermore, since the data arrivals can be sustained by the energy harvesting process, the SSBP-EH policy tries to get as close a the non-EH one, leading to the small of the gap. Also, note that the SBP-EH and SSBP-EH policies are more volatile than their non-EH counterparts. For example, around $t=550$, the number of queued packets spikes for the energy harvesting policies, which is not the case in the non-EH ones. These types of spikes arise due to a certain lack of energy around those time instants.

\begin{figure}[t!]
    \centering
    \subfigure[SBP-EH.]
    {
    	\includegraphics[scale=0.95]{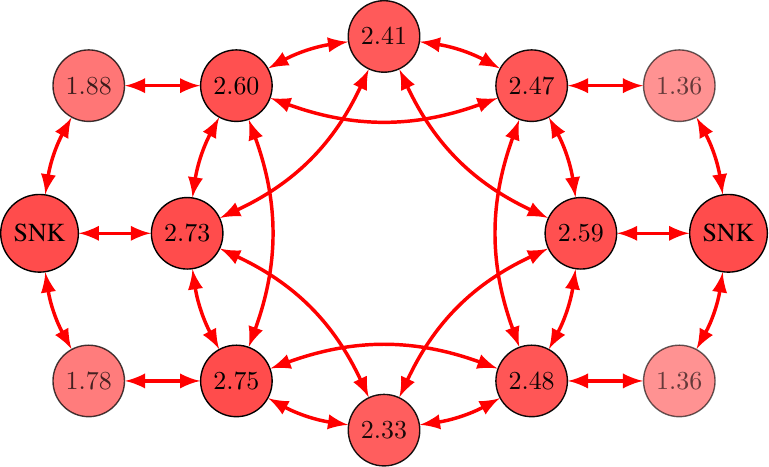}
        \label{fig:queuesAtNodesSBP}
    }
    \subfigure[SSBP-EH.]
    {
 	   \includegraphics[scale=0.95]{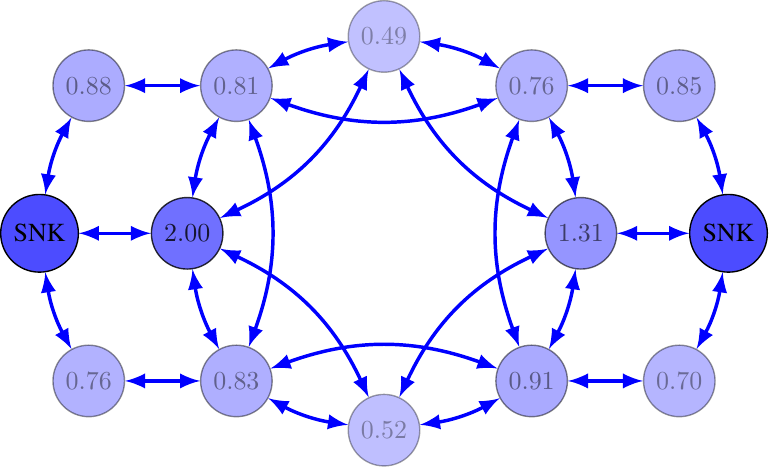}
        \label{fig:queuesAtNodesSSBP}
    }
    \caption{Average data queues at each node in the network.}
    \label{fig:queuesAtNodes}
\end{figure}

In Figure \ref{fig:queuesAtNodes} we have plotted the average queued packets at each node for the SBP-EH and SSBP-EH policies. In general, SSBP-EH shows a lower number of average queued packets over all the nodes and the improvements are more significant the lower the pressure the node supports. This tends to translate to better improvements for nodes far away from a sink that tend to be routed less traffic. For example, the nodes $7$ and $8$ (See Fig. \ref{fig:Network}), which are the furthest away from any sink, show a reduction of $1.92$ and $1.81$ average packets, respectively, when using SSBP-EH. The rest of the nodes also show significant improvements when using SSBP-EH. Nodes $5$, $6$, $9$ and $10$, all lying at two hops of distance of a sink are more critical for accessing a sink, as having them congested blocks the access to the sink of the previous nodes $7$ and $8$. In this case, the improvements range from $1.57$ to $1.79$ average data packets. Finally, there are the nodes that lie at one hop distance from any sink (nodes $2$, $3$, $4$, $11$, $12$ and $13$) . These nodes sustain a significant amount of traffic and show improvements ranging from $0.51$ to $1.78$. With the nodes with the highest traffic, nodes $4$ and $11$,  improving by $0.73$ and $1.28$ data packets, respectively.

The differences between SBP-EH and SSBP-EH are also evidenced in terms of their energy use. In Figure \ref{fig:NetworkEnergy} we plot the total energy in the network at a given time slot for both the SBP-EH and the SSBP-EH policies. On one hand, this figure illustrates the high variability in the energy supply due to the energy harvesting process. On the other hand, the SSBP-EH policy is shown to be more aggressive in its energy use. Also, note that drops in total network energy are not necessarily correlated with increases in queued packets in the network. For example, the previously noticed peak of queued data packets at $t=550$ in Fig. \ref{fig:QueuedPackets} does not have an equivalent large drop in network energy. This is due to the fact that it is better for energy in the network to overall be lower than to have a specific high-pressure node have an energy shortage. In general, spikes in queued data packets tend to occur when a specific route becomes blocked by the temporary lack of energy. 

\begin{figure}[t]
	\centering
    \includegraphics[width=\columnwidth]{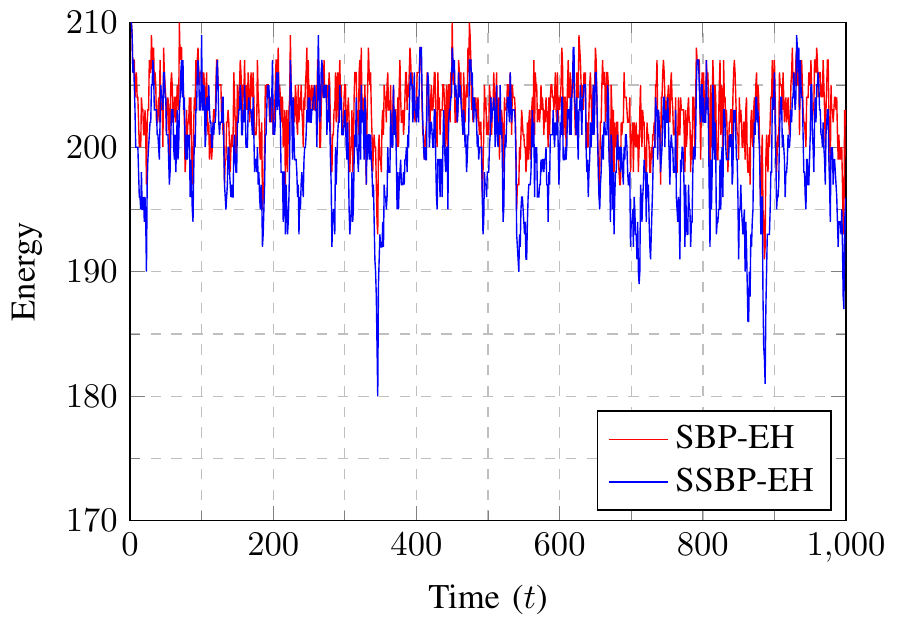}
	\caption{Total energy stored in the network at a given time slot.}
	\label{fig:NetworkEnergy}
\end{figure}

\subsection{Network Balance}

\begin{figure}[t]
	\centering
    \includegraphics[width=\columnwidth]{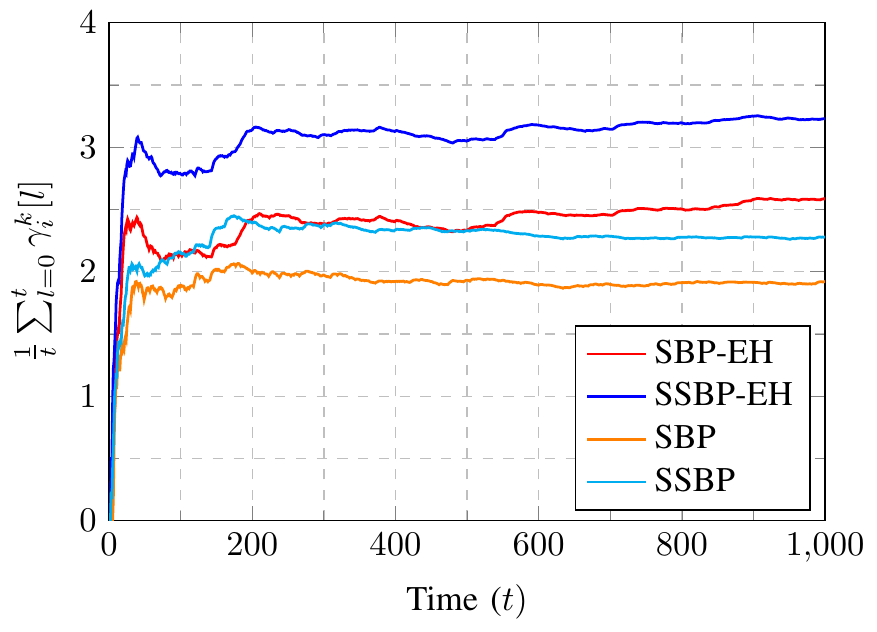}
	\caption{Average value of the queue multipliers $\gamma_i^k$ for node 5 ($i=5$, $k=1$).}
	\label{fig:AvgQueueDual}
\end{figure}

As discussed in Section \ref{sec:Stability}, the choice of the parameters $\bar{\gamma}_i^{k}$, which control the maximum values taken by the queue multipliers $\gamma_i^{k}$, is important to ensure the stability of the data queues. Namely, the optimal multipliers must be smaller than this $\bar{\gamma}_i^{k}$ parameter. In Figure \ref{fig:AvgQueueDual}, we plot the $\gamma_i^k$ multipliers for one of the nodes which supports the most traffic in the network (node 5). The time-average of these dual variables converges to the optimal value. In the chosen scenario, the parameter used, $\bar{\gamma}_i^{k}=10$, is well above the optimal value. Hence, the system satisfies Proposition \ref{cor:QueueStability}, and can be ensured to stabilize the queues. Some additional insight into the importance of the queue multipliers can be gained by a pricing interpretation of the dual problem. Under this interpretation, the dual variables $\gamma_i^{k}[t]$ represent the unit price associated to the routing constraint $a^k_i[t] \leq \sum_{j \in \mathcal{N}_i} \bigl(r_{ij}^k[t] - r_{ji}^k[t]\bigr)+x_i^k[t]$. When the node does not satisfy this constraint, it pays $\gamma_i^{k}[t]$ per unit of constraint violation. Likewise, if it strictly satisfies this constraint, it receives $\gamma_i^{k}[t]$ per unit of constraint satisfaction. In this sense, the $\bar{\gamma}_i^{k}$ parameter represents both the maximum payment that a node can receive and the maximum price it can pay. Hence, the optimal value of $\gamma_i^k$ must necessarily fall below $\bar{\gamma}_i^{k}$ in order to obtain a stable system. We can use this pricing interpretation to compare the different policies. In general, the energy harvesting policies have higher $\gamma_i^{k}[t]$ values than their non-EH counterparts. This is due to the fact that, due to the energy harvesting constraints, the unit violation of the routing constraint is harder to recoup in the EH-aware policies, hence the higher price paid. In a similar note, due to their more aggressive routing decisions, the soft policies also show higher $\gamma_i^{k}[t]$ values than their non-soft counterparts.

Also of interest is the study of the balance characteristics of the network. As discussed previously, the stability guarantees of the network are subject to the existence of a feasible routing solution given the data and energy arrival rates. This motivates another way of showing stability, different from the data queues shown in Fig. \ref{fig:QueuedPackets}. We can consider that a successful routing strategy is expected to route to the sink nodes as many packets as generated by the network. This is given by the network balance expression $\sum_{i \in \mathcal{N}}\sum_{k \in \mathcal{K}}\left( a^k_i[t] -  r_{ij}^k[t] \right)$, where $j=N^k_{(dest)}$. The time average of this measure is shown in Figure \ref{fig:NetworkBalance}. As expected, the time average data network balance goes to zero for all policies. This illustrates that all policies are capable of routing to the sink nodes as many packets as they arrive to the network, hence ensuring queue stability. We previously observed in Fig. \ref{fig:NetworkBalance} that stability occurs around $t=200$ time slots for the SBP and SBP-EH policies and less than $t=100$ time slots for the SSBP and SSBP-EH ones. Those observations can be compared with the network balance of Fig. \ref{fig:NetworkBalance}, where those values correspond to the time around when the slope of the data balance curve starts to go flat. Remarkably, the proposed energy harvesting policies do not lose convergence speed when compared to the non-EH ones. Also, convergence of the SSBP and SSBP-EH policies occurs at a faster rate, a point that we previously raised in Section \ref{subsec:SSBP-EH}.

\begin{figure}[t]
	\centering
    \includegraphics[width=\columnwidth]{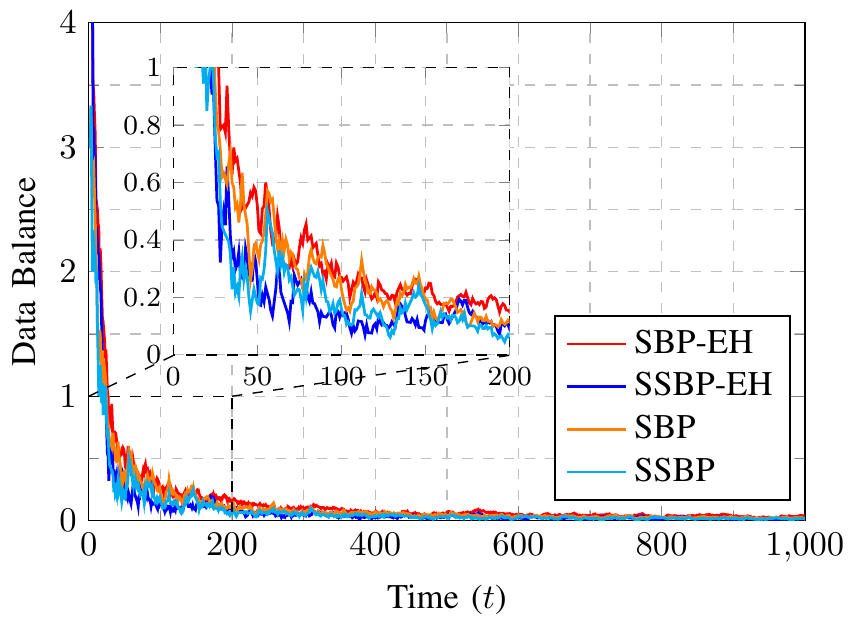}
	\caption{Average data balance in the network, given by the expression $\frac{1}{t}\sum_{l=0}^{t}\sum_{i \in \mathcal{N}}\sum_{k \in \mathcal{K}}\left( a^k_i[l] -  r_{ij}^k[l] \right)$, where $j=N^k_{(dest)}$.}
	\label{fig:NetworkBalance}
\end{figure}

\begin{figure}[t]
	\centering
    \includegraphics[width=\columnwidth]{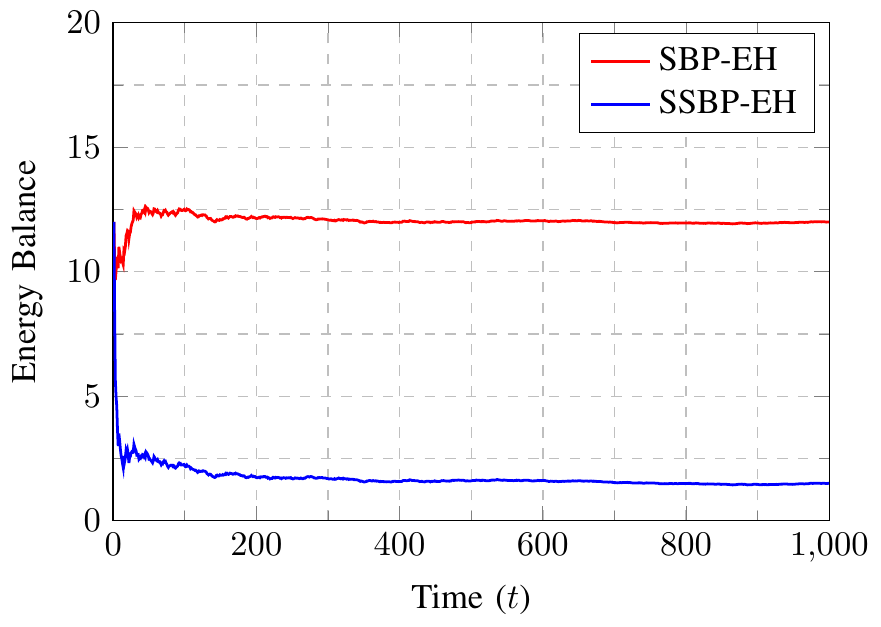}
	\caption{Average energy balance in the network, given by the expression $\frac{1}{t}\sum_{l=0}^{t}\sum_{i \in \mathcal{N}}\sum_{k \in \mathcal{K}}\left( e_i[l] -  \sum_{j \in \mathcal{N}_i} r_{ij}^k[l] \right)$.}
	\label{fig:EnergyBalance}
\end{figure}

Another measure of network balance of interest is related to the energy balance in the network. This can be expressed by $\sum_{i \in \mathcal{N}}\sum_{k \in \mathcal{K}}\left( e_i[t] -  \sum_{j \in \mathcal{N}_i} r_{ij}^k[t] \right)$. This measure serves to quantify how much of the energy harvested in the network is actually being used. The time average of the energy balance is shown in Figure \ref{fig:EnergyBalance}. As expected, given that the network harvests enough energy to support the routing-scheduling decisions, both policies converge to a non-zero value. Once stabilized, the SBP-EH policy has, in average, energy left for around 12 packet transmission in all of the network, while the SSBP-EH only has energy left for an average of 2 packet transmissions. We previously identified in Fig. \ref{fig:NetworkEnergy} the SSBP-EH  to be more aggressive in its energy use. At the same time, we can also say that the SSBP-EH policy uses its energy supply in a more efficient manner. Since the nodes are powered by energy harvesting instead of a limited energy supply, not using available energy can be considered wasteful, as batteries will tend to overflow. In this sense, to use more energy (as in SSBP-EH) rather than to use energy more conservatively (as in SBP-EH), can be seen as a better option. In this sense, SSBP-EH makes a more efficient use of the available energy, resulting in an overall better performance.

\subsection{Network Delay}

An additional important characteristic of routing-scheduling policies is their resulting delay in the packet delivery. While the average delay is proportional to the average number of queued packets in the network, we also study this measure explicitly. In order to do this, and under the assumption of first-in first-out queues, we compute the number of time slots it takes for a packet to be delivered to a sink node. We plot in Figure \ref{fig:HistogramTS} the resulting histogram. In average, the number of time slots it takes to deliver a packet to a sink node is $4.04$ for the SSBP-EH policy, while it is $5.36$ for the SBP-EH policy. This is about a $1$ time slot of difference between the policies. Taking a more detailed look at the histogram, we can see that  the distribution for the SSBP-EH is very similar to the one of the SBP-EH, but with a $1$ time slot shift to the left. As already seen in Fig. \ref{fig:queuesAtNodes}, the more aggressive behavior of the SSBP-EH policy leads to an overall reduction in the network queues. These smaller queues result in a reduction of the waiting time of packets at each hop, which results in a smaller delivery delay.

\begin{figure}[t]
	\centering
    \includegraphics[width=\columnwidth]{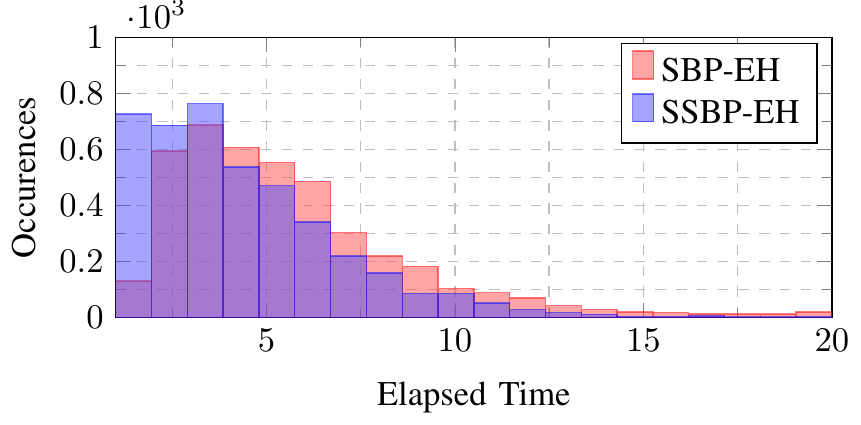}    
	\caption{Histogram of packet elapsed time before reaching a sink node.} 
	\label{fig:HistogramTS}
\end{figure}

\section{Conclusions}
\label{sec:Conclusions}

In this work, we have generalized the stochastic family of backpressure policies to energy harvesting networks. Different from other works, which are based on Lyapunov drift notions, we have resorted to duality theory. This has allowed us to study the problem under a framework based on the correspondence between queues and Lagrange multipliers. Under this framework, we have proposed two policies, (i) SBP-EH, an easy to implement policy where nodes track the difference between their queue multipliers and the ones of their neighbors. The pressure is further reduced by the battery multipliers as the stored energy decreases. Then, the transmit decision is to transmit the flow with the highest pressure. And (ii) SSBP-EH, a probabilistic policy  with improved performance and convergence guarantees, where nodes track the pressure in the same way as SBP-EH but perform an equalization in the form of an inverse waterfilling. This results in a probability mass function for the routing-scheduling decision, where a sample of this distribution is then taken to decide the transmission. For both policies, we have studied the conditions under which energy causality and queue stability are guaranteed, which we have also verified by means of simulations. The numerical results show that given feasible data and energy arrivals, both policies are capable of stabilizing the network. Overall, the SSBP-EH policy shows improvements in queued packets, stabilization speed and delay with respect to the SBP-EH policy. Furthermore, when compared to non-EH policies, the SSBP-EH policy shows to have an asymptotically vanishing gap.

\bibliographystyle{IEEEtran}
\bibliography{bib}

\end{document}